\documentclass{llncs}
\usepackage[utf8]{inputenc}
\usepackage[english]{babel}
\usepackage{amsmath, amssymb}
\usepackage{graphicx}
\usepackage{color}
\usepackage{xspace}							%%% Hezké mezery po použití maker na náhrady slov 

\def\IT{\it\aftergroup\/}

\newcommand{\decproblem}[3]{\begin{tabular}{ll}
{\bf Problem:}\hspace*{1em} & {\sc #1}\\ {\IT Instance}: & #2\\ {\IT Question:} & #3
\end{tabular}}

\newcommand\I{$\mathcal{I}$}

\newcommand\V{$\mathcal{V}$}

\newcommand\IV{{$\mathcal{I}$\kern .05em$\mathcal{V}$}\xspace}
\newcommand\IVproblem{{\sc \IV-matching}\xspace}
\newcommand\threeDproblem{{\sc 3D-matching}\xspace}
\newcommand\NP{{\sf NP}\xspace}

\usepackage[pdftex,unicode,hidelinks]{hyperref}
\hypersetup{breaklinks=true}
\hypersetup{pdftitle={IV-matching is strongly NP-hard}}
\hypersetup{pdfauthor={Lukáš Folwarczný and Dušan Knop}}
\hypersetup{pdfkeywords={IV-matching, perfect matching, NP-completeness}}

\begin{document}

\title{\IV-matching is strongly \NP-hard}

\author{Lukáš Folwarczný\inst{1}
\and
Dušan Knop\inst{2}\fnmsep\thanks{Author supported by the project SVV--2015--260223, by the project CE-ITI P202/12/G061 of GA ČR and by the GA UK project 1784214.}}

\institute{
Computer Science Institute,\\
Faculty of Mathematics and Physics,\\
Charles University in Prague, Czech Republic\\
\email{folwar@iuuk.mff.cuni.cz}
\and
Department of Applied Mathematics,\\
Faculty of Mathematics and Physics,\\
Charles University in Prague, Czech Republic\\
\email{knop@kam.mff.cuni.cz}
}

\maketitle

\begin{abstract}
\IV-matching is a~generalization of perfect bipartite matching. The complexity of finding
\IV-matching in a~graph was posted as an open problem at the ICALP 2014 conference.

In this note, we resolve the question and prove that, contrary to the expectations of the
authors, the given problem is strongly \NP-hard (already in the simplest non-trivial case
of four layers). Hence it is unlikely that there would be an efficient (polynomial or
pseudo-polynomial) algorithm solving the problem.

\end{abstract}

\section{Introduction}
The perfect matching problem is one of the central and well studied problems in graph theory.
For an exhaustive overview of this area of research we refer to the monograph of Lovász and Plummer~\cite{lovasz-plummer}. An important part of research is devoted to special classes of graphs (especially to bipartite graphs) as well as various strengthening of the former problem (for example T-joins).
For the bipartite setting special attention is focused on min--max type theorems as is for example the pioneering work of König~\cite{konig}.

The complexity status of the \IVproblem problem was posted as an open problem in the full version of
the article by Fiala, Klavík, Kratochvíl and Ne\-de\-la~\cite{klavik} at the ICALP~2014
conference. \IV-matching is a~generalization of perfect bipartite matching to
multi-layered graphs.  Informally, there is a~classical one-to-one matching
between layers $2k-1$ and $2k$, but a~``two-to-one matching'' between layers $2k$
and $2k+1$. There are further restrictions imposed on the matching.

\paragraph{Motivation.}
Fiala et al.~\cite{klavik} studied algorithmic aspects of regular graph covers
(a~graph~$G$ covers a~graph~$H$ if there exists a~locally bijective homomorphism from~$G$
to~$H$; regular covers are very symmetric covers). The problem of determining for two
given graphs~$G$ and $H$ whether $G$ regularly covers~$H$ generalizes both the graph
isomorphism problem and the problem of Cayley graphs~\cite{cayley} recognition. The need to determine
whether there is an~\IV-matching in a~given graph appears naturally in
an algorithm proposed by the authors. Therefore, they ask whether there is an
efficient algorithm for this task.

\paragraph{Our contribution.}
We show that the problem is strongly \NP-complete, already in the simplest non-trivial
case of four layers, and therefore a~polynomial or pseudo-polynomial algorithm is unlikely.
Thus our results imply a dichotomy for the complexity of the \IVproblem problem.
The concepts around \NP{} are explained for example by Arora and
Barak~\cite{arora-barak}.

\paragraph{Outline of the note.} A~detailed problem definition is given in
Section~\ref{sec:definition} and the proof of \NP-completeness is given in
Section~\ref{sec:np-completeness}.

\section{Problem Definition}
\label{sec:definition}

The formal definition of \IV-matching is as follows.

\begin{definition}[Layered graph]
Let $G=(V,E)$ be a~bipartite graph. We say that $G$ is a~{\em layered graph} if $G$ fulfills
the following conditions:
\begin{itemize}
\item Vertices are partitioned into $\ell$~sets $V_1, \dots, V_\ell$ called {\em layers}. There
are only edges between two consecutive layers $V_k$ and $V_{k+1}$ for $k=1,\dots,\ell-1$.
\item Each layer is further partitioned into {\em clusters}.
\item Edges of~$G$ are described by edges on clusters; we call these edges {\em
macroedges}. If there is a~macroedge between two clusters, then vertices of these two
clusters induce a~complete bipartite graph. If there is no macroedge, then these vertices
induce an edge-less graph.
\item Macroedges between clusters of layers $V_{2k}$ and $V_{2k+1}$ form a~matching (not
necessarily a~maximum matching).
\item There are no conditions for macroedges between clusters of the layers $V_{2k-1}$ and $V_{2k}$.
\end{itemize}
\end{definition}

\begin{definition}[\IV-matching]
Let $G$ be a layered graph.
{\em \IV-matching} is a~subset of edges $M\subseteq E$ such that: Each vertex from an even
layer $V_{2k}$ is incident to exactly one vertex from $V_{2k-1}\cup V_{2k+1}$.
Each vertex from the layer $V_{2k-1}$ is either incident to exactly two vertices from
$V_{2k-2}$, or exactly one vertex from $V_{2k}$ (these two options are exclusive).

It implies that edges between layers $V_{2k-1}$ and $V_{2k}$ form a~matching (\I-shapes)
and edges between layers $V_{2k}$ and $V_{2k+1}$ form independent \V-shapes with centers
in the layer $V_{2k+1}$.
\end{definition}

\begin{figure}
\centering
\includegraphics{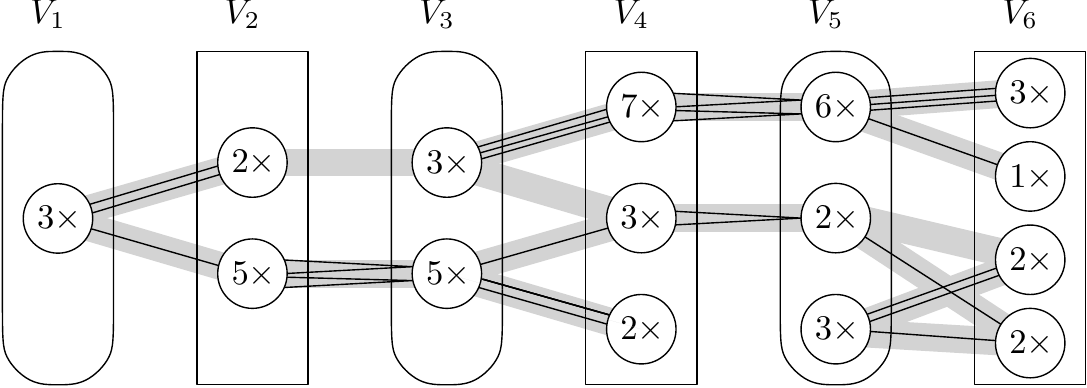}
\caption{Layered graph with $6$ layers. Numbers denote the number of vertices in a cluster.
Macroedges are represented by shaded edges, while edges form an \IV-matching.}
\end{figure}

As the \IVproblem problem we denote the decision problem of finding out whether there is an
\IV-matching in a~given graph.
\vskip 0.5cm
\decproblem{\IVproblem}
{Graph $G$ described by clusters and macroedges.}
{Is there an \IV-matching in $G$?}
\vskip 0.5cm

For $\ell = 2$ the problem is just an ordinary bipartite matching. For odd values
of $\ell$, the clusters of the last layer $V_\ell$ can be matched in only one possible way,
thus this odd case reduces to the case with $(\ell-1)$ layers. The first interesting case
is therefore $\ell = 4$ and we prove in the next section that this case is already \NP-complete.

\section{\NP-completeness}
\label{sec:np-completeness}

We prove \NP-completeness of the \IVproblem problem using a~reduction from the 
\threeDproblem{} problem  in this section.

\paragraph{Three-dimensional matching.} In the \threeDproblem\ problem we are given
a~hypergraph $H = (U,F)$. The hypergraph is tripartite, i.e., the set of vertices is
partitioned into three equally sized partites $X$, $Y$ and $Z$. Each hyperedge consists of exactly
one vertex from each partite, thus $F \subseteq X \times Y \times Z$.  A~set of pairwise
disjoint hyperedges covering all vertices is called a {\IT perfect matching}.

\vskip 0.5cm
\decproblem{\threeDproblem}
{A~tripartite hypergraph $H$.}
{Is there a~perfect matching in $H$?}
\vskip 0.5cm

The \threeDproblem\ problem is well-known to be \NP-complete; it is actually the seventeenth problem
in the Karp's set of 21 \NP-complete problems~\cite{karp}.

\begin{theorem}
The \IVproblem\ problem is strongly \NP-complete, already in the case of $\ell = 4$.
\end{theorem}

\begin{proof}
It is easy to see that the problem is in the class \NP{}: The \IV-matching itself is a~polynomial
certificate and its correctness can be directly verified.

In order to show that the problem is \NP-hard, we construct a~polynomial-time reduction
from the \threeDproblem{} problem to the \IVproblem problem.
Let the hypergraph $H = (U, F)$ be an instance of the
\threeDproblem\ problem with partites $X$, $Y$, $Z$ and let us write 
$n = |X| = |Y| = |Z|$ and $m= |F|.$

We construct the instance~$G = (V,E)$ of the \IVproblem\ problem as follows. We put vertices from $X$
and $Y$ into the layer~$V_1$ and vertices from $Z$ into the layer~$V_4$. Each vertex forms
its own cluster of the size one. Then for each hyperedge $e=\langle x,y,z\rangle$ we add
a~cluster with two new vertices $x_e, y_e$ into the layer $V_2$ and a~cluster with one new
vertex $z_e$ into the layer $V_3$. We then add these four edges on clusters: $\langle
\{x\}, \{x_e,y_e\} \rangle$, $\langle\{y\}, \{x_e, y_e\}\rangle$, $\langle\{x_e, y_e\},
\{z_e\}\rangle$ and $\langle\{z_e\},\{z\}\rangle$.

The key idea of our construction is that \V{}s in an \IV-matching solution translate to hyperedges {\IT
not} present in the perfect matching. An example is given in Fig.~\ref{fig-reduction}.

\begin{figure}
\centering
\includegraphics{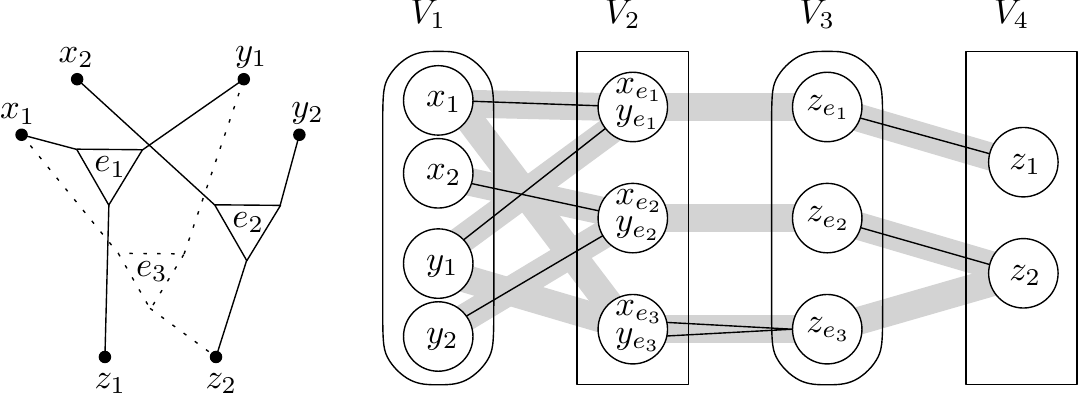}
\caption{An instance of the \threeDproblem\ problem with a~perfect matching and an equivalent instance
of the \IVproblem\ problem instance with the corresponding \IV-matching.}
\label{fig-reduction}
\end{figure}

The resulting instance of the \IVproblem\ problem has $3n+3m$ vertices, $4m$~edges and it can be
constructed directly in polynomial time. We shall now prove that there is an~\IV-matching
in $G$ if and only if there is a~perfect matching in the original hypergraph $H$.

$\Rightarrow$ Let $M$ be an \IV-matching in $G$. Observe that for each hyperedge $e \in F$
necessarily either all vertices $x_e$, $y_e$ and $z_e$ are matched with \I s or all
of them are matched with a~single \V.

Let us put into our matching of~$H$ all hyperedges $e \in F$ such that $x_e$, $y_e$ and
$z_e$ are matched with \I s. Because $M$ is an \IV-matching, every vertex of the original
hypergraph is connected by an \I\ to exactly one vertex in $V_2 \cup V_3$ and we chose the
corresponding edge to our matching so all vertices of the hypergraph are covered. Because
there are $3n$ vertices in $V_1 \cup V_4$ and $3m$ vertices in $V_2 \cup V_3$, the
number of \V s used is $m-n$ and so we used $n$~hyperedges in our matching of the
hypergraph. This proves that we constructed a~perfect matching.

$\Leftarrow$ Let $N \subseteq F$ be a~perfect matching in the hypergraph~$H$. For each
hyperedge $e = \langle x,y,z\rangle$ we connect by \I s the pairs of vertices $\{x, x_e\}$, $\{y,
y_e\}$ and $\{z, z_e\}$. For each hyperedge $e \notin N$ we cover the vertices $x_e, y_e, z_e$
by a~\V. Note that every vertex in $V_2 \cup V_3$ is covered by exactly one
\I\ or one \V.

Because the matching~$N$ covers all vertices in~$H$, every vertex in $V_1 \cup V_4$ is covered
by at least one~\I. Because we put $3n$ \I s into the graph, every vertex is covered by
exactly one~\I. This implies that we obtained a~correct \IV-matching.\\

This finishes the proof of \NP-completeness.
Because we only use clusters of size two in the reduction, the problem stays naturally
\NP-complete when input is encoded in unary. Therefore, strong \NP-completeness is
established as well.
\qed
\end{proof}

\subsubsection*{Acknowledgements.}
We would like to thank Pavel Klavík for bringing the \IVproblem problem to our
attention and a~fruitful discussion.  Moreover, we wish to express our sincere gratitude
to the organizers of the summer REU program 2014 at the Center for Discrete Mathematics
and Theoretical Computer Science (DIMACS), Rutgers University where the main part of our
research was conducted.

\bibliographystyle{splncs03}
\bibliography{iv-matching}

\end{document}